\documentclass[runningheads]{llncs}

\usepackage{comment}
\usepackage[whole,autotilde]{bxcjkjatype}
\usepackage{amsmath,amssymb}
\usepackage{algorithm,algpseudocode}
\usepackage{booktabs}
\usepackage[sort]{cite}
\usepackage{bxascmac}
\usepackage[T1]{fontenc}

\usepackage{tikz}
\usetikzlibrary{positioning}
  
\newcommand{\ceil}[1]{\lceil #1 \rceil}
\newcommand{\floor}[1]{\lfloor #1 \rfloor}

\usepackage{color}
\usepackage{CJKutf8}
\usepackage{ifthen}
\newcommand{\COMM}[2]{{
\begin{CJK}{UTF8}{ipxm}
\ifthenelse{\equal{#1}{SK}}{\color{blue}}{
\ifthenelse{\equal{#1}{TM}}{\color{red}}{
\ifthenelse{\equal{#1}{RS}}{\color{cyan}}{
\ifthenelse{\equal{#1}{BB}}{\color{magenta}}}}}
[#1: #2]
\end{CJK}
}}
\begin{document}

\title{Linear Pseudo-Polynomial Factor Algorithm for Automaton Constrained Tree Knapsack Problem}

\titlerunning{Constrained tree knapsack problem} %optional, please use if title is longer than one line

\author{Soh Kumabe\inst{1,2} \and Takanori Maehara\inst{2} \and Ryoma Sin'ya\inst{3}}

%}{The University of Tokyo, Japan \\ RIKEN Center for Advanced Intelligence Project, Japan}{sou.kumabe@riken.jp}{}{}

%\author{Takanori Maehara}{RIKEN Center for Advanced Intelligence Project, Japan}{takanori.maehara@riken.jp}{}{}

%\author{Ryoma Sin'ya}{Akita University, Japan}{ryoma@math.akita-u.ac.jp}{}{}

\authorrunning{S. Kumabe et al.} %mandatory. First: Use abbreviated first/middle names. Second (only in severe cases): Use first author plus 'et al.'

%\author{First Author\inst{1}\orcidID{0000-1111-2222-3333} \and
%Second Author\inst{2,3}\orcidID{1111-2222-3333-4444} \and
%Third Author\inst{3}\orcidID{2222--3333-4444-5555}}
\institute{% 
The University of Tokyo, Tokyo, Japan \\ 
\email{sou.kumabe@riken.jp} \and 
RIKEN Center for Advanced Intelligence Project, Tokyo, Japan \\
\email{takanori.maehara@riken.jp} \and 
Akita University, Akita, Japan \\
\email{ryoma@math.akita-u.ac.jp}
}
\maketitle              % typeset the header of the contribution
\begin{abstract}
The \emph{automaton constrained tree knapsack problem} is a variant of the knapsack problem in which the items are associated with the vertices of the tree, and we can select a subset of items that is accepted by a tree automaton.
If the capacities or the profits of items are integers, the problem can be solved in pseudo-polynomial time using the dynamic programming algorithm.
However, this algorithm has a quadratic pseudo-polynomial factor in its complexity because of the max-plus convolution.
In this study, we propose a new dynamic programming technique, called \emph{heavy-light recursive dynamic programming}, to obtain algorithms having linear pseudo-polynomial factors in the complexity.
Such algorithms can be used for solving the problems with polynomially small capacities/profits efficiently, and used for deriving efficient fully polynomial-time approximation schemes.
We also consider the $k$-subtree version problem that finds $k$ disjoint subtrees and a solution in each subtree that maximizes total profit under a budget constraint. We show that this problem can be solved in almost the same complexity as the original problem.

 %mandatory
\end{abstract}

%\author{John Q. Public}{Dummy University Computing Laboratory, [Address], Country}{johnqpublic@dummyuni.org}{https://orcid.org/0000-0002-1825-0097}{[funding]}%mandatory, please use full name; only 1 author per \author macro; first two parameters are mandatory, other parameters can be empty.
%\author{Joan R. Public}{Department of Informatics, Dummy College, [Address], Country}{joanrpublic@dummycollege.org}{[orcid]}{[funding]}

%\subjclass{
%Theory of computation $\rightarrow$ Discrete optimization,
%Theory of computation $\rightarrow$ Dynamic programming,
%Theory of computation $\rightarrow$ Tree languages
%}% mandatory: Please choose ACM 2012 classifications from https://www.acm.org/publications/class-2012 or https://dl.acm.org/ccs/ccs_flat.cfm . E.g., cite as "General and reference $\rightarrow$ General literature" or \ccsdesc[100]{General and reference~General literature}. 

%\ccsdesc[500]{Theory of computation~Discrete optimization}
%\ccsdesc[500]{Theory of computation~Dynamic programming}
%\ccsdesc[300]{Theory of computation~Tree languages}

\section{Introduction}

\subsection{Background and Motivation}

The knapsack problem seeks a set of items that maximizes total profit under a budget constraint. 
The problem is one of the most fundamental combinatorial optimization problems~\cite{kellerer2003knapsack} and has many real-world applications such as scheduling~\cite{ibarra1975fast}, network design~\cite{van2006tree}, and natural language processing~\cite{hirao2013single}.
The problem is NP-hard; however, if the profits or the weights of items are integers, the problem can be solved using the dynamic programming (DP) that runs in pseudo-polynomial time.
This algorithm is the basis for the fully-polynomial time approximation scheme (FPTAS) of the knapsack problem~\cite{ibarra1975fast,lawler1977fast}.

Here, we consider the \emph{automaton constrained tree knapsack problem}, which is defined as follows.
Let $T = (V(T), E(T))$ be a rooted tree 
where $V(T)$ is the set of vertices and $E(T)$ is the set of edges, $\mathcal{F}(\mathcal{A}) \subseteq 2^{V(T)}$ be a feasible domain represented by a top-down tree automaton (see Section~\ref{sec:treeautomaton} for details).
We denote by $n = |V(T)|$ the number of vertices in $T$.
Each $u \in V(T)$ has profit $p(u) \in \mathbb{R}_{\ge 0}$ and weight $w(u) \in \mathbb{R}_{\ge 0}$.
For a vertex subset $X \subseteq V(T)$, we define $p(X) = \sum_{u \in X} p(u)$ and $w(X) = \sum_{u \in X} w(u)$.
Let $C \in \mathbb{R}_{\ge 0}$ be the capacity.
Then, the task is to solve the following optimization problem: 
\begin{align}
\label{eq:treeknapsack}
\text{maximize} \ \ p(X) \ \ \text{subject to} \ \ w(X) \le C, \ \ X \in \mathcal{F}(\mathcal{A}),
\end{align}
This is a quite general problem since any constraint on a tree specified by a monadic second-order logic formula is represented by a tree automaton~\cite{thatcher1968generalized}.
For example, the precedence constrained problem~\cite{lukes1974efficient}, the connectivity constrained problem~\cite{hochbaum1994node}, and the independent set constrained problem~\cite{pferschy2009knapsack} are particular cases of this problem (See Examples~\ref{ex:ind}, \ref{ex:prec}, and \ref{ex:conn}).

As in the case of the standard knapsack problem, the automaton constrained tree knapsack problem can be solved by DP.
If the tree automaton has a \emph{polynomially bounded diversity of transitions} (see Section~\ref{sec:treeautomaton} for the definition), 
the complexity of the algorithm is $O(\text{poly}(n) C^2)$ time if the weights are integers, and $O(\text{poly}(n) P^2)$ time if the profits are integers, where $P$ is an upper bound of the optimal value (see Section~\ref{sec:standardDP}).
Several existing studies have considered particular cases of the problem and derived the corresponding realization of this algorithm~\cite{lukes1974efficient,hochbaum1994node,pferschy2009knapsack}.
%Examples include Lukes's algorithm for the precedence constrained problem~\cite{lukes1974efficient}, Hochbaum and Pathria's algorithm for the connectivity constrained problem~\cite{hochbaum1994node}, and Pferschy and Schauer's algorithm for the independent set constrained problem~\cite{pferschy2009knapsack}.

In this study, we focus on the \emph{pseudo-polynomial factors} $C$ or $P$ in the complexity.
The quadratic pseudo-polynomial factors of the standard DP come from merging solutions to the subtrees, which is implemented by the max-plus (or min-plus) convolution, whose current best complexity is $O(N^2 \log \log N / \log^2 N)$, where $N$ is the length of the arrays~\cite{bremner2006necklaces}. 
It is conjectured that the max-plus convolution requires $\Omega(N^{2 - \delta})$ time for any $\delta > 0$~\cite{cygan2017problems,backurs2017better,bremner2006necklaces}. 
However, quadratic pseudo-polynomial factors are sometimes unacceptable. 
For example, in practice, we often encounter the case that $C$ is polynomially greater than $n$ (e.g., $n = 100$ and $C = 100,000$). 
In this case, quadratic pseudo-polynomial factors are not desirable.
For another example, when we derive a FPTAS from the DP, we take $P \propto 1/\epsilon$; thus, a smaller degree in $P$ implies a faster algorithm with the same accuracy.
The purpose of this study is to derive algorithms for the problem that run in $O(\text{poly}(n) C)$ or $O(\text{poly}(n) P)$ time.

Thus far, the only studies that have addressed this issue are those on the precedence constrained knapsack problem.
Johnson and Niemi~\cite{johnson1983knapsacks} proposed a technique, called left-right DP, which runs in $O(n C)$ time.
Cho and Shaw~\cite{cho1997depth} proposed a variant of the left-right DP, called depth-first DP, which also runs in $O(n C)$ time.
However, we do not know what kinds of constraints (other than the precedence constraint) admit algorithms with complexity that is linear in pseudo-polynomial factors.

\subsection{Our Contribution}

In this study, we introduce a new DP technique, called \emph{heavy-light recursive dynamic programming (HLRecDP)}. 
This technique is motivated by Chekuri and Pal's recursive greedy algorithm for the $s$-$t$ path constrained monotone submodular maximization problem~\cite{chekuri2005recursive} and its generalization to the logic constrained monotone submodular maximization problem~\cite{tomas2018algorithmic}.
It also generalizes the left-right DP and depth-first DP for precedence constrained problem to the automaton constrained problem.
Formally, by using this technique, we obtain the following theorem.
From now on, we denote the logarithm of base two by $\log$.
\begin{theorem}
\label{thm:hlrecdp}
Let $T = (V(T), E(T))$ be a tree with $n$ vertices and $\mathcal{A}$ be a non-deterministic top-down tree automaton with the diversity of transitions $\delta(n)$.
Let $p(i) \in \mathbb{R}_{\ge 0}$, $w(i) \in \mathbb{Z}_{\ge 0}$, and $C \in \mathbb{Z}_{\ge 0}$.
Then, there is an algorithm for problem~\eqref{eq:treeknapsack} that runs in $O(n^{\log (1 + \delta(n))} C)$ time.
In particular, if $\delta(n) = O(1)$, the algorithm runs in $O(\text{poly}(n) C)$ time.%
\footnote{For simplicity, we only consider the case in which the weights are integers. The same result is obtained when the profits are integers.}
\end{theorem}
This theorem gives a sufficient condition for admitting (pseudo-)polynomial time algorithms with linear pseudo-polynomial factors. 
By applying this theorem to the precedence constrained problem, we obtain $O(n C)$ time algorithm that is is equivalent to the existing left-right DP ~\cite{johnson1983knapsacks} and depth-first DP~\cite{cho1997depth} (Example~\ref{ex:prec}).
%If we apply this theorem to the precedence constrained problem, we obtain an $O(n C)$ time algorithm, which is equivalent 

We then consider the \emph{$k$-subtree} version problem.
Let $k = O(1)$ be an integer. 
Then, the problem is to find $k$ disjoint subtrees of the given tree and a feasible solution in each subtree such that the total profit is maximized under the total budget constraint. 
For example, the $k$ connected component constrained problem is the $k$-subtree version of the precedence constrained problem.
By using the property of the algorithm of Theorem~\ref{thm:hlrecdp} and divide-and-conquer techniques, we show that this problem can be solved in almost the same time complexity as the original problem.
\begin{theorem}
\label{thm:ksubtree}
Suppose that $\mathcal{A}$ is a prefix-closed top-down tree automaton with the bounded diversity of transitions, and the automaton constrained tree knapsack problem with $\mathcal{A}$ can be solved in $f(n)$ time by Algorithm~\ref{alg:hlrecdp}.
Let $k = O(1)$.
Then, there exists an algorithm for the corresponding $k$-subtree version problem that runs in the following complexity:
\begin{description}
\item [$k = 1$.] $O(f(n) \log n)$ if $f(n) = O(n C)$, and $O(f(n))$ time if $f(n) = O(n^{e} C)$ for some $e > 1$.
%and $O(f(n) \log n)$ time if $f(n) = O(n C)$%; the hidden constants do not depend on $e$.
\item [$k \ge 2$.] $O(f(n) (\log n)^{\log k})$ time if $f(n) = O(n^{e} C)$ for some $e > 1$; the hidden constant is a polynomial in $k$.
%Here, the constant factor is a polynomial in $k$.
\end{description}
\end{theorem}
This theorem gives an $O(n \log n C)$ time algorithm for the connectivity constrained problem, and an $O(n^{e} C)$ time algorithm for any $e > 1$ for the $k$ connected component constrained tree knapsack problem.

\subsection*{Organization of the Paper}

The paper is organized as follows.
In Section~\ref{sec:treeautomaton}, we introduce top-down tree automata. 
In Section~\ref{sec:standardDP}, we introduce the standard DP using a top-down tree automaton. 
In Section~\ref{sec:recdp}, we prove Theorem~\ref{thm:hlrecdp} by introducing the HLRecDP.
In Section~\ref{sec:accel}, we prove Theorem~\ref{thm:ksubtree} using the divide-conquer technique with HLRecDP.

\section{Preliminaries}

\subsection{Tree Automaton}
\label{sec:treeautomaton}

A \emph{non-deterministic top-down tree automaton (``automaton'' for short)}~\cite{comon2017tree} is a tuple $\mathcal{A} = (Q, \Sigma, Q_\text{init}, \Delta)$, where $Q$ is the set of states, $\Sigma$ is a set of alphabets, $Q_\text{init} \subseteq Q$ is the set of initial states, and $\Delta$ is a set of rewriting rules of the form
\begin{align}
    Q \times \Sigma \ni (q, \sigma) \mapsto (q_1, \ldots, q_d) \in Q \times \cdots \times Q.
\end{align}
We assume that the number of states of the automaton is constant, $|Q| = O(1)$.
The automaton is \emph{prefix-closed} if $(q, \sigma) \mapsto (q_1, \ldots, q_d)$ is in $\Delta$ then $(q, \sigma) \mapsto (q_1, \ldots, q_{d-1})$ also in $\Delta$. 

The \emph{run} of the automaton is defined as follows.
Let $T = (V(T), E(T))$ be a rooted tree, and $\sigma: V(T) \to \Sigma$ be labels on the vertices.
The automaton first assigns an initial state $q \in Q_\text{init}$ to the root of the tree.
Then it processes the tree from the top (root) to the bottom (leaves).
If vertex $u \in V(T)$ has state $q \in Q$, we choose a rewriting rule $(q, \sigma(u)) \mapsto (q_1, \ldots, q_d)$ and assign the states $q_1, \ldots, q_d$ to the children $v_1, \ldots, v_d \in V(T)$ of $u$, respectively.
Note that, if no rule is applicable to $u$ and $q$, the run fails. 
The automaton accepts a labeled tree if there is at least one run from the root to the leaves in which the state of the root is in $Q_\text{init}$.

To represent a substructure of a tree using an automaton, we choose the alphabet $\Sigma = \{0, 1\}$ and identify the subgraph $X \subseteq V(T)$ as the labels $\sigma_X: V(T) \to \Sigma$ such that $\sigma_X(u) = 1$ for $u \in X$ and $\sigma_X(u) = 0$ for $u \not \in X$.
Then, the family of subsets $\mathcal{F}(\mathcal{A}) \subseteq 2^{V(T)}$ represented by this automaton is specified by
\begin{align}
    \mathcal{F}(\mathcal{A}) = \{ X \subseteq V(T) : \text{$\mathcal{A}$ accepts $T$ with label $\sigma_X$} \}.
\end{align}

To evaluate the complexity of DP, we introduce the following quantity $\delta(n)$, called the \emph{diversity of transitions}.
\begin{align}
    \delta(n) = \max_{m \le n} |\bigcup_{\text{``}(q,\sigma) \mapsto (q_1, \ldots, q_m)\text{''} \in \Delta} \{ (q_1, \ldots, q_m) \}|.
\end{align}
By definition, $\delta(n)$ is monotone in $n$.
Intuitively, $\delta(n)$ is the maximum number of subproblems in DP; see Section~\ref{sec:standardDP} below.
There is an automaton with exponentially large diversity of transitions, i.e., $\delta(n) = \Theta(|Q|^n)$, and in such case, it looks impossible to obtain $O(\text{poly}(n))$ time algorithm. 
Therefore, we assume some boundedness of $\delta(n)$.

\subsection{Quadratic Pseudo-Polynomial Factor Algorithm}
\label{sec:standardDP}

Here, we introduce the standard DP that solves the problem in $O(\text{poly}(n) C^2)$ time if the automaton has a polynomially bounded diversity of transitions~\cite{lukes1974efficient,hochbaum1994node,pferschy2009knapsack}.
We regard this as a baseline algorithm for the problem.

Let $T = (V(T), E(T))$ be a rooted tree. 
We denote by $T_u$ the subtree of $T$ rooted by $u \in V(T)$.
The algorithm computes array $x_{u,q}$ of length $C + 1$ for each $u \in V(T)$ and $q \in Q$, such that
\begin{align}
\label{eq:xa}
    x_{u,q}[c] = \max \{ & p(X) : X \subseteq V(T_u), w(X) = c,\ \text{subtree $T_u$ with labels $\sigma_X$ is} \notag \\ 
    & \text{accepted by $\mathcal{A}$, where the initial state is $q$} \}.
\end{align}
Once the array for the root vertex $r \in V(T)$ is obtained, the optimal value is computed by $\max_{q \in Q_\text{init}, c \in \{0, \ldots, C\}} x_{r,q}[c]$ in $O(|Q_\text{init}| C) = O(C)$ time.

We compute these arrays using the bottom-up DP as follows.
For each leaf, the array is immediately computed in $O(\delta(0) C) = O(C)$ time. 
Consider a vertex $u \in V(T)$ with children $v_1, \ldots, v_d \in V(T)$, such that the arrays $x_{v,q}$ are computed for all $v \in \{v_1, \ldots, v_d\}$ and $q \in Q$.
Then, 
\begin{align}
    x_{u,q}[c] = \max \{ & x_{v_1,q_1}[c_1] + \cdots + x_{v_d,q_d}[c_d] + w(u) \sigma : \notag \\ & {(q,\sigma) \mapsto (q_1,\ldots,q_d) \in \Delta}, c_1 + \cdots + c_d + w(u) \sigma = c \} 
\end{align}
%\begin{align}
%    x_{u,q}[c] = \max_{(q,\sigma) \mapsto (q_1,\ldots,q_d) \in \Delta} \max_{\substack{c_1 + \cdots + c_d + w(u) \sigma = c}} x_{v_1,q_1}[c_1] + \cdots + x_{v_d,q_d}[c_d] + w(u) \sigma.
%\end{align}
The maximization with respect to $c_1, \ldots, c_d$ is evaluated by the max-plus convolution; thus, it costs about $O(n C^2)$ time.
For the maximization with respect to $(q,\sigma) \to (q_1, \ldots, q_d) \in \Delta$, we only have to evaluate the formula for distinct $(q_1, \ldots, q_d)$.
Therefore, the complexity of evaluating \eqref{eq:xa} is $O(n \delta(n) C^2)$ time, and the total complexity is $O(n^2 \delta(n) C^2) = O(\text{poly}(n) C^2)$.

\section{Heavy-Light Recursive Dynamic Programming}
\label{sec:recdp}

In this section, we present the HLRecDP for obtaining an $O(n^{\log (1 + \delta(n))} C)$ time algorithm. 
In Section~\ref{sec:recdpbal}, we first propose the recursive dynamic programming (RecDP) technique for balanced trees. 
%If the tree has height $O(\log n)$, this technique provides the desired time complexity.
To handle non-balanced trees, in Section~\ref{sec:hlrecdp}, we combine the heavy-light decomposition to the RecDP.
%Indeed, these sections provide the proof of Theorem~\ref{thm:hlrecdp}.

\subsection{Recursive Dynamic Programming for Balanced Trees}
\label{sec:recdpbal}

Our goal is to compute arrays $\{ x_{r,q} \}_{q \in Q_\text{init}}$ for the root $r \in V(T)$ of the tree, where $x_{r,q}$ is defined in \eqref{eq:xa}.
To avoid quadratic pseudo-polynomial factors, we call the recursive procedure for the children multiple times, instead of merging subtree solutions.

Formally, we design procedure $\textsc{RecDP}(u, q, x)$, where $u \in V(T)$, $q \in Q$, and $x$ is an array of size $C+1$.
It computes array $y_{u,q,x}$ defined by %of length $C+1$ defined by
\begin{align}
\label{eq:ya}
    y_{u,q,x}[c] = \max \{ & p(X) + x[c']: X \subseteq V(T_u), w(X) + c' = c,\ \text{subtree $T_u$ with} \notag \\ 
    & \text{labels $\sigma_X$ is accepted by $\mathcal{A}$, where the initial state is $q$} \}.
\end{align}
The difference between \eqref{eq:xa} and \eqref{eq:ya} is that \eqref{eq:ya} contains the array parameter $x$, which corresponds to the ``initial values'' of the DP.
More intuitively, it returns an array that is obtained by ``adding'' items in the subtree $T_u$ optimally to the current solution represented by $x$.
By calling $\textsc{RecDP}(r, q, [0,-\infty,\ldots,-\infty])$, where $r \in V(T)$ is the root of the tree and $q \in Q_\text{init}$, we obtain the desired solution $x_{r,q}$.

Here, $\textsc{RecDP}(u, q, x)$ is implemented as follows.
If $u \in V(T)$ is a leaf, we can compute \eqref{eq:ya} in $O(C)$ time.
Consider a vertex $u \in V(T)$ that has children $v_1, \ldots, v_d \in V(T)$. 
For each rewriting rule $(q, \sigma) \mapsto (q_1, \ldots, q_d)$, we first call $\textsc{RecDP}(v_1,q_1,x)$ to obtain array $y_1 = y_{v_1,q_1,x}$. 
Then, we call $\textsc{RecDP}(v_2, q_2, y_1)$ to obtain array $y_2 = y_{v_2,q_2,y_1}$, i.e., we use the returned array $y_1$ as the initial values of the DP to the subtree rooted by $v_2$. 
By iterating this process to the last child, we obtain array $y_d = y_{v_d,q_d,y_{d-1}}$.
The solution corresponds to this rewriting rule is then obtained by 
\begin{align}
    z_{u,q,x}[c] = y_d[c - \sigma w(u)] + \sigma p(u), \quad c \in \{0, \ldots, C\}.
\end{align}
By taking the entry-wise maximum of the solutions on all of the rewriting rules, we obtain the solution to $\textsc{RecDP}(u, q, x)$.

The correctness of the above procedure is easily checked.
We evaluate the time complexity.
Let $f(n)$ be the complexity of the procedure. Let $n_1, \ldots, n_d$ be the number of vertices on the subtrees rooted by $v_1, \ldots, v_d$.
Then we have $n_1 + \cdots + n_d = n - 1$.
Because the algorithm calls the procedure recursively to each subtree at most $\delta(n)$ times, the complexity satisfies
\begin{align}
    f(n) \le \delta(n) (f(n_1) + \cdots + f(n_d)) + O(C).\footnotemark
\end{align}
\footnotetext{The additive term is naturally $O(d C)$; however, it is separated and included in the recursive terms.}%
If the tree is balanced, i.e., $n_j \le n/2$ for all $j = 1, \ldots, d$, this already provides the desired complexity: 
Without loss of generality, we can assume that $f(n)$ is convex in $n$.
Then, the maximum of the right-hand side is attained at $n_1 = \ceil{(n-1)/2}$, $n_2 = \floor{(n-1)/2}$, and $n_3 = \cdots = n_d = 0$.
Therefore,
\begin{align}
    f(n) \le \delta(n) \left( f(\ceil{(n-1)/2}) + f(\floor{(n-1)/2}) \right) + O(C).
\end{align}
By solving this inequality, we have $f(n) = O((2 \delta(n))^{\log n} C) = O(n^{1 + \log \delta(n)} C)$.
% n^c C = (2 delta) 2 T(n/2) = (2 delta) 2 (n/2)^c C + O(d C)
%       = (2 delta) 2 / 2^c n^c C + O(d C)
% (4 delta)/2^c = 1 ==> c = log(4 delta)

\subsection{Heavy-Light Recursive Dynamic Programming}
\label{sec:hlrecdp}

To obtain an $O(\text{pseudopoly}(n) C)$ time algorithm for general (i.e., non-balanced) trees, we have to make the depth of the recursion to $O(\log n)$.
The HLRecDP achieves this by using the \emph{heavy-light decomposition}~\cite{sleator1983data}.

First, we introduce the heavy-light decomposition.
Let $T = (V(T), E(T))$ be a rooted tree whose edges are directed toward the leaves.
An edge $(u, v) \in E(T)$ is a \emph{heavy edge} if $v$ has more descendants than other children do (ties are broken arbitrary).
An edge is a \emph{light edge} if it is not a heavy edge.
A vertex $v \in V(T)$ is a \emph{heavy child} of $u \in V(T)$ if $(u, v)$ is a heavy edge. A \emph{light child} is defined similarly.
A subtree rooted by a light child is referred to as a \emph{light subtree}.
The set of heavy edges forms disjoint paths, called \emph{heavy paths}. 
The tree is decomposed into the heavy paths, which is referred to as a \emph{heavy-light decomposition}.
The most important property of a heavy-light decomposition is that, for each $u \in V(T)$, the number of descendants of a light child is at most $|V(T_u)| / 2$.

Recall algorithm $\textsc{RecDP}(u, q, x)$ defined in the previous section. 
We observe that all recursive calls of $\textsc{RecDP}(v_1, q_1, x)$ to the first child has the same initial array $x$ for different $q_1$.
Thus, we can ``gather'' all recursive calls for the first child into a single recursive call.
The HLRecDP sets the heavy child as the first child to avoid an excessive number of recursive calls this child.

Formally, we define procedure $\textsc{HLRecDP}(u, x)$. 
This returns a set of arrays $\{ y_{u,q} \}_{q \in Q}$, where $y_{u,q}$ is defined in \eqref{eq:ya}.
For $v_1, \ldots, v_d \in V(T)$ and $q_1, \ldots, q_d \in Q$, we define $\textsc{HLRecDP}(v_1, \ldots, v_d, x)_{q_1, \ldots, q_d}$ as a shorthand notation of the sequential evaluation
\begin{align}
\textsc{HLRecDP}(v_d, \textsc{HLRecDP}(v_{d-1} \cdots \textsc{HLRecDP}(v_1, x)_{q_1} \cdots )_{q_{d-1}})_{q_d}.
\end{align}
Now we describe the procedure. 
Let $v_1, \ldots, v_d \in V(T)$ be the children of $u$, where $v_1$ is the heavy child.
First, we call $\textsc{HLRecDP}(v_1, x)$ and store the resulting arrays for all $q \in Q$.
Then, for each rewriting rule $(q,\sigma) \mapsto (q_1, \ldots, q_d)$, we call $\textsc{HLRecDP}(v_2, \ldots, v_d, \textsc{HLRecDP}(v_1, x)_{q_1})_{q_2, \ldots, q_d}$ and add item $u$ if $\sigma = 1$ to obtain the solution to the rewriting rule.
By taking the entry-wise maximum over the rewriting rules, we obtain the desired solution.
See Algorithm~\ref{alg:hlrecdp} for the detailed procedure.

By construction, $\textsc{HLRecDP}$ gives the same solution as $\textsc{RecDP}$; thus, it correctly solves the problem.
We evaluate the complexity as follows.
Let $n_1, \ldots, n_d$ be the number of vertices on the subtrees rooted by $v_1, \ldots, v_d$.
As same as the analysis of $\textsc{RecDP}$, the complexity $f(n)$ of the algorithm satisfies the following inequality.
\begin{align}
    f(n) \le f(n_1) + \delta(n) \left(f(n_2) + \cdots + f(n_d)\right) + O(C).
\end{align}
By the convexity of $f(n)$ and the heavy-light property, i.e., $n_j \le n/2$ ($j = 2, \ldots, d$), the maximum of the right-hand side is attained at $n_1 = \ceil{(n-1)/2}$, $n_2 = \floor{(n-1)/2}$, and $n_3 = \cdots = n_d = 0$. Thus, we have
\begin{align}
    f(n) \le f(\ceil{(n-1)/2}) + \delta(n)) f(\floor{(n-1)/2}) + O(C).
\end{align}
By solving this inequality, we have $f(n) = O(n^{\log (1 + \delta(n))} C)$.
%This proves Theorem~\ref{thm:hlrecdp}.  
\qed

\begin{remark}
There is a gap of the tractable classes between the standard DP (Section~\ref{sec:standardDP}) and the HLRecDP.
The analysis in Section~\ref{sec:standardDP} implies that we can obtain $O(\text{poly}(n) C^2)$ time algorithm if $\delta(n)$ is polynomially bounded.
On the other hand, the analysis in this section implies that if $\delta(n)$ is polynomially bounded (rather than bounded by a constant), we can only obtain an algorithm with quasi-polynomial time complexity, i.e., $n^{O(\log n)} C$.
%If the input tree has a bounded degree, Algorithm~\ref{alg:hlrecdp} runs in $O(\text{poly}(n) C)$ time regardless of $\delta(n)$.
\end{remark}

\begin{algorithm}[tb]
\caption{Heavy-Light Recursive Dynamic Programming}
\label{alg:hlrecdp}
\begin{algorithmic}[1]
\Procedure{HLRecDP}{$u$, $x$}
\State{$y_{u,q}[c] = -\infty$ for all $c \in \{0, \ldots, C\}$, $q \in Q$}
\State{Let $v_1, \ldots, v_d$ be the children of $u$, where $v_1$ is the heavy child}
\State{Call $\textsc{HLRecDP}(v_1, x)$ and store the arrays for $q \in Q$}
\For{``$(q, \sigma) \mapsto (q_1, \ldots, q_d)$'' $\in \Delta$}
\State{Let $z = \textsc{HLRecDP}(v_2, \ldots, v_d, \textsc{HLRecDP}(v_1, x)_{q_1})_{q_2, \ldots, q_d}$}
\For{$c = 0, \ldots, C$}
\State{$y_{u,q}[c] \leftarrow \max \{ y_{u,q}[c], z[c - \sigma w(u)] + \sigma p(u) \}$}
\EndFor
\EndFor
\State{\textbf{return} $\{ y_{u,q} \}_{q \in Q}$}
\EndProcedure
\end{algorithmic}
\end{algorithm}

Here, we derive several results for particular cases using our method.

\begin{example}[Independent Set Constrained Problem]
\label{ex:ind}
Let us consider the \emph{independent set constrained} tree knapsack problem whose feasible set contains no adjacent vertices.
This constraint is represented by an automaton $\mathcal{A} = (Q, \Sigma, Q_\text{init}, \Delta)$, where $Q = Q_\text{init} = \{ \mathrm{s}, \mathrm{x} \}$ and
\begin{align}
    (\mathrm{s}, 0) \mapsto (\mathrm{s}, \ldots, \mathrm{s}), \quad 
    (\mathrm{s}, 1) \mapsto (\mathrm{x}, \ldots, \mathrm{x}), \quad 
    (\mathrm{x}, 0) \mapsto (\mathrm{s}, \ldots, \mathrm{s}).
\end{align}
Here, state $\mathrm{s}$ means the vertex can be selected and state $\mathrm{x}$ means the vertex cannot be selected.
The diversity of transitions is $\delta(n) = 2$ because the rules for $(\mathrm{s}, 0)$ and $(\mathrm{x},0)$ have the same right-hand side; therefore, we can solve the independent set constrained tree knapsack problem in $O(n^{\log (1 + \delta(n))} C) = O(n^{\log 3} C) = O(n^{1.585} C)$ time. 
\end{example}

\begin{example}[Precedence Constrained Problem]
\label{ex:prec}
Let us consider the \emph{precedence constrained} tree knapsack problem whose feasible set is precedence closed, i.e., if a vertex is contained in a solution, all the precedences are also contained in the solution.
This constraint is represented by an automaton $\mathcal{A} = (Q, \Sigma, Q_\text{init}, \Delta)$, where $Q = Q_\text{init} = \{ \mathrm{s}, \mathrm{x} \}$ and
\begin{align}
    (\mathrm{s}, 0) \mapsto (\mathrm{x}, \ldots, \mathrm{x}), \quad 
    (\mathrm{s}, 1) \mapsto (\mathrm{s}, \ldots, \mathrm{s}), \quad 
    (\mathrm{x}, 0) \mapsto (\mathrm{x}, \ldots, \mathrm{x}).
\end{align}
Here, state $\mathrm{s}$ means the vertex can be selected and state $\mathrm{x}$ means the vertex cannot be selected.
Since the diversity of transitions is $\delta(n) = 2$, the algorithm runs in $O(n^{\log (1+\delta(n))} C) = O(n^{1.585} C)$ time.

This complexity can be improved further. 
If a vertex has state $\mathrm{x}$, we cannot select all of the descendants of the vertex; thus we obtain the solution for this case without calling the procedure recursively. 
Thus, the required number of recursive calls is at most one, which is for $(\mathrm{s}, 1)$.
Therefore, the algorithm runs in $O(n^{\log (1+1)} C) = O(n C)$ time.
Note that this algorithm ``coincides'' with the left-right DP~\cite{johnson1983knapsacks} and the depth-first DP~\cite{cho1997depth} in the sense that these perform the same manipulations. %in the same processing order.
\end{example}

\begin{example}[Connectivity Constrained Problem]
\label{ex:conn}
Let us consider the \emph{connectivity constrained} tree knapsack problem whose feasible set forms a connected subgraph of a given tree.
This constraint is represented by an automaton $\mathcal{A} = (Q, \Sigma, Q_\text{init}, \Delta)$, such that $Q = \{ \mathrm{s}, \mathrm{o}, \mathrm{x} \}$, $Q_\text{init} = \{ \mathrm{s} \}$ and
\begin{align}
    & (\mathrm{s}, 0) \mapsto (\mathrm{s}, \mathrm{x}, \ldots, \mathrm{x}), \quad 
    (\mathrm{s}, 0) \mapsto (\mathrm{x}, \mathrm{s}, \ldots, \mathrm{x}), \quad \ldots, \quad 
    (\mathrm{s}, 0) \mapsto (\mathrm{x}, \mathrm{x}, \ldots, s), \notag \\ 
    & (\mathrm{s}, 1) \mapsto (\mathrm{o}, \mathrm{o}, \ldots, \mathrm{o}), \quad
    (\mathrm{o}, 0) \mapsto (\mathrm{x}, \mathrm{x}, \ldots, \mathrm{x}), \quad 
    (\mathrm{o}, 1) \mapsto (\mathrm{o}, \mathrm{o}, \ldots, \mathrm{o}), \notag \\ 
    & (\mathrm{x}, 0) \mapsto (\mathrm{x}, \mathrm{x}, \ldots, \mathrm{x}).
\end{align}
Here, state $\mathrm{s}$ means the vertex can be selected, state $\mathrm{o}$ means the  vertex is now selecting, and state $\mathrm{x}$ means that the vertex cannot be selected.
Note that $\mathcal{A}$ is non-deterministic because there are $d$ rules for $(\mathrm{s}, 0)$. 
Thus, the diversity of transitions is $\delta(n) = n$, which is not bounded by a constant. 
Thus, the theorem gives only quasi-polynomial time algorithm.

To improve the performance, we make the similar observation to the precedence constraint (Example~\ref{ex:prec}).
Then, the number of recursive calls to each subtree is at most twice; one is for $(\mathrm{s}, 0)$ and the other is for $(\mathrm{s}, 1)$ and $(\mathrm{o}, 1)$. 
Therefore, the algorithm runs in $O(n^{\log (1 + 2)} C) = O(n^{1.585} C)$ time.
\end{example}

\begin{example}[$k$ Connected Component Constrained Problem]
\label{ex:kconn}
Let us consider \emph{$k$ connected component constrained} tree knapsack problem whose feasible solution is $k$ connected components. 
By using the same technique as the connectivity constrained problem (Example~\ref{ex:conn}), we obtain $n^{O(\log k)} C$ time algorithm for the problem.
Note that, if we handle $k$ as a kind of weight, we can derive $O(k n^e C) = O(n^e C)$ time algorithm for some universal constant $e$.
\end{example}

\section{k-Subtree Version Problems}
\label{sec:accel}

In this section, we consider the $k$-subtree version problems and prove Theorem~\ref{thm:ksubtree}.
We introduce two auxiliary problems: 
The first one is the \emph{for-all-subtree} problem that requires to solve the problem on each subtree $T_u$ of $T$ rooted by $u \in V(T)$.
The second one is the \emph{for-all-subtree-complement} problem that requires to solve the problem on each subtree-complement $T \setminus T_u$ of $T$ for all $u \in V(T)$.
These problems can be solved in almost the same time complexity as follows.
\begin{lemma}
\label{lem:forallsubtree}
Suppose that the automaton constrained tree knapsack problem with tree automaton $\mathcal{A}$ can be solved in $f(n) = O(n^e C)$ time by Algorithm~\ref{alg:hlrecdp}.
Then, the corresponding for-all-subtree version problem can be solved in $O(f(n) \log n)$ time if $e = 1$ and $O(f(n))$ time if $e > 1$.
%for some $e > 0$ and $O(f(n) \log n)$ if $f(n) = O(n (\log n)^e C)$ for $e \ge 0$.
\end{lemma}
\begin{proof}
Let us fix a heavy path $u_1, \ldots, u_l$ that starts from the root of the tree, i.e., $u_1$ is the root and $u_l$ is a leaf.
First, we call $\textsc{HLRecDP}(u_1, [0, -\infty, \ldots, -\infty])$ to the root $u_1$ of the tree.
Then, it recursively calls $\textsc{HLRecDP}(u_i, [0, -\infty, \ldots, -\infty])$ to the the vertices $u_2, \ldots, u_l$ on the heavy path.
This means that this single call gives the solutions to the subtrees rooted by the vertices on the heavy path.

After this computation, we call the procedure recursively to the light subtrees adjacent to the heavy path to solve the problem.
The total complexity $g(n)$ satisfies
\begin{align}
	g(n) \le g(n_1) + \cdots + g(n_s) + f(n), 
\end{align}
where $n_1, \ldots, n_s$ are the sizes of the subtrees. 
By definition, $n_1 + \cdots + n_s \le n-1$.
Also, by the heavy-light property, $n_j \le n/2$ ($j = 1, \ldots, s$). 
Therefore, the maximum of the right-hand side is attained at $n_1 = \ceil{(n-1)/2}$, $n_2 = \floor{(n-1)/2}$, and $n_3 = \cdots = n_s = 0$.
Thus,
\begin{align}
	g(n) \le g(\ceil{(n-1)/2}) + g(\floor{(n-1)/2}) + f(n),
\end{align}
By solving this inequality we obtain the desired result.
\qed
\end{proof}

\begin{lemma}
\label{lem:forallsubtreecomplement}
Suppose that the automaton constrained tree knapsack problem with tree automaton $\mathcal{A}$ can be solved in $f(n) = O(n^e C)$ time by Algorithm~\ref{alg:hlrecdp}.
%Then, the corresponding for-all-subtree-complement version problem can be solved in $O(f(n) \log n)$ if $f(n) = O(n^{1+e} C)$ for $e > 0$ and $O(f(n) (\log n)^2)$ time if $f(n) = O(n (\log n)^e C)$ for $e \ge 0$.
Then, the corresponding for-all-subtree-complement version problem can be solved in $O(f(n) (\log n)^2)$ time if $e = 1$ and $O(f(n) \log n)$ time if $e > 1$.
\end{lemma}
\begin{proof}
For vertex $u \in V(T)$, we define array $x_{u,q}$ of length $C+1$ that represents the solution on $T \setminus T_u$, where the parent of $u$ has state $q \in Q$.
We compute the arrays for all the vertices.
We define $x_{r,q} = [0, -\infty, \ldots, -\infty]$ for the root $r \in V(T)$ and all $q \in Q$. 
Let us fix a heavy path $u_1, \ldots, u_l$ that starts from the root of the tree.
We compute the arrays for the vertices on the heavy path, and for the vertices adjacent to the heavy path separately.

\medskip

\noindent \textbf{Vertices on the heavy path.} \ 
Suppose that we have $\{ x_{u_{i-1}, q} \}_{q \in Q}$. 
%We compute $\{ x_{u_i, q} \}_{q \in Q}$ for all $i = 2, \ldots, l$ from the arrays of the root $\{ x_{u_1,q} \}_{q \in Q}$.  
Let $v_1, \ldots, v_d$ be children of $u_{i-1}$, where $v_1 = u_i$.
For each rewriting rule $(q, \sigma) \mapsto (q_1, \ldots, q_{d-1}) \in \Delta$, which is a rule of length $d-1$, which will match to $v_2, \ldots, v_d$, the array corresponds to this rule is obtained by calling $\textsc{HLRecDP}(v_2, \ldots, v_d, x_{u_{i-1},q})_{q_1,\ldots,q_{d-1}}$ and by adding $u_i$ if $\sigma = 1$.
By taking the entry-wise maximum of the arrays for different rules, we obtain $\{ x_{u_i, q} \}_{q \in Q}$.
Since this computation process pays the same computational effort as $\textsc{HLRecDP}(u_1, x)$, the complexity is $f(n)$.

\medskip

\noindent \textbf{Vertices adjacent to the heavy path.} \ 
We compute $\{ x_{v, q} \}_{q \in Q}$ for all light child $v$ adjacent to the heavy path.
It is obtained by calling \textsc{HLRecDP} to all the subtrees except $T_v$; however, this method involves redundant computations. 
We reduce the complexity by storing intermediate results by a segment tree-like divide-and-conquer technique.

First, we compute arrays $y_{i, j, q_i, q_j}$ for $i = 0, \ldots, l-1$, $j = i+1, \ldots, l$, and $q_i, q_j \in Q$. 
This stores the vector obtained by calling $\textsc{HLRecDP}$ with initial array $x_{u_1, q}$ for some $q$ to the subtree except the light children of $u_{i+1}, \ldots, u_j$, where the states of $u_i$ and $u_i$ are $q_i$ and $q_j$, respectively.
Initially, we set $y_{0, l, q_0, q_l} = x_{u_1, q_0}$ for all $q_0, q_l \in Q$.
If we have $\{ y_{i, j, q_i, q_j} \}_{q_i, q_j \in Q}$ for $i+1 < j$, we can compute $\{ y_{i, m, q_i, q_m} \}_{q_i, q_m \in Q}$ where $m = \floor{(i+j)/2}$ by calling $\textsc{HLRecDP}$ to the light subtrees of $u_{m+1}, \ldots, u_j$ with initial array $z_{i, j, q_i, q_j}$.
Similarly, we can compute $\{ y_{m, j, q_{m}, q_r}\}_{q_{m}, q_r \in Q}$. % where $m = \floor{(i + j) / 2}$. 
The complexity of computing all the arrays is $O(f(n) \log n)$ since $\textsc{HLRecDP}$ is called to subtree $T_{u_i}$ at most $O(\log n)$ times.

Next, for each $u_k$ ($k = 1, \ldots, l-1$) on the heavy path, we consider the children $v_1, \ldots, v_d$ of $u_k$, where $v_1$ is the heavy child (i.e., $v_1 = u_{k+1}$). 
For each rewriting rule $(q, \sigma) \mapsto (q_1, \ldots, q_{d-1}) \in \Delta$, we compute arrays $z_{i,j,q,q_1}$ for $i = 1, \ldots, d-1$ and $j = i+1, \ldots, d$.
This stores the vector obtained by calling $\textsc{HLRecDP}$ with initial array $y_{k-1, k, q, q_1}$ to the subtrees except $v_{i+1}, \ldots, v_j$, and is computed by the same technique as $y$. 
Once the arrays are obtained, we can retrieve $x_{v_i,q}$ by taking the entry-wise maximum of $z_{i-1,i,q,q_1}$ with respect to $q_1$. 
Thus, the total complexity of this part is $O(f(n) \log n)$.

After this computation, we call the procedure recursively to the light subtrees adjacent to the heavy path to solve the problem.
The total complexity $g(n)$ satisfies 
\begin{align}
	g(n) \le g(n_1) + \cdots + g(n_s) + O(f(n) \log n),
\end{align}
where $n_1 + \cdots + n_s \le n-1$ and $n_j \le n/2$ ($j = 1, \ldots, s)$.
By solving this inequality as similar to Lemma~\ref{lem:forallsubtree}, we obtain the desired result.
\qed
\end{proof}

Now we provide an outline of the proof of Theorem~\ref{thm:ksubtree}.
We combines the algorithms for the for-all-subtree version problem and the for-all-subtree-complement version problem. % with the divide-and-conquer technique as follows.
\begin{proof}[of Theorem~\ref{thm:ksubtree}, outline]
We design algorithm $\textsc{Conn}(u, k, x)$ that computes arrays $x_{u,q,b,l}$ where $u \in V(T)$, $q \in Q$, $b \in \{0, 1\}$, and $l \in \{0, \ldots, k\}$. 
The array represents the solution to the subtree $T_u$ such that the root ($= u$) has state $q$ and is included by a subtree if $b = 1$, and $l$ subtrees are selected.
If $k = 0$, the solution is $[0, -\infty, \ldots, -\infty]$. 
If $k = 1$, we can solve the problem by solving for-all-subtree version problem since the automaton is prefix closed.
Thus, in the following, we consider $k \ge 2$.

Let $g(n, k)$ be the complexity of the algorithm for $n$ vertices with parameter $k$. We derive the recursive relation of $g$. 
We fix a heavy path, and consider light subtrees adjacent to the heavy path. 

\noindent \textbf{Case 1: There is a light subtree $T_{v}$ that contains at least $k/2$ components.} \quad 
In this case, the subtree complement $T_u \setminus T_{v}$ contains at most $k/2$ components. 
Thus, we guess such subtree $T_v$ and solve the problem on $T_u \setminus T_v$ and $T_v$ separately.
We can solve all the subtree complements simultaneously by calling the subtree-complement version of $\textsc{Conn}(u, k/2, *)$. 
Also, we can solve each subtree by calling $\textsc{Conn}(v, k, *)$.
The complexity of this approach is $g(n_1, k) + \cdots + g(n_s, k) + O(g(n, k/2) \log n)$. 

%Thus, we first find such solutions $z_{v}$ for all $v$ by applying the for-all-subtree-complement algorithm for 
%$\textsc{Conn}(u, k/2, x)$.
%Then, we call the procedure $\textsc{Conn}(v, k, z_v)$ to obtain solution to $T_{v}$.

\medskip

\noindent \textbf{Case 2: Otherwise; i.e., all the light subtrees contain at most $k/2$ components.} \quad 
We call $\textsc{Conn}(v, k/2, *)$ for all subtrees $v$, sequentially.
The complexity of this part is given by $g(n_1, k/2) + \cdots + g(n_s, k/2) \le g(n, k/2)$.

\medskip 

The total complexity $g(n, k)$ of the algorithm satisfies
\begin{align}
    g(n, k) \le g(n_1, k) + \cdots + g(n_s, k) + O( g(n, k/2) \log n ). 
    %\notag \\
    %& + g(n_1, k/2) + \cdots + g(n_s, k/2).
\end{align}
By using $n_1 + \cdots n_s \le n-1$ and $n_j \le n/2$ ($j = 1, \ldots, s$), we obtain $g(n, k) \le h(k) f(n) (\log n)^{\log k}$, where $h(k)$ is a polynomial in $k$.
\qed
\begin{comment} 

Let us consider a vertex $u \in V(T)$ and its children $v_1, \ldots, v_d \in V(T)$, where $v_1$ is the heavy child of $u$.
First, by calling the procedure recursively to the heavy child, we obtain solution $\{ x_{v_1,q,b,l} \}_{q,b,l}$.
Next, we split into two cases and handle each of them.

\medskip

\noindent \textbf{Case 1: There is a light subtree $T_{v_i}$ that contains at least $k/2$ components.} \quad 
In this case, the subtree complement $T_u \setminus T_{v_i}$ contains at most $k/2$ components. 
Thus, we first find such solution $z_{v_j}$ by applying the for-all-subtree-complement algorithm for $\textsc{Conn}(u, k/2, x)$.
Then, we call the procedure $\textsc{Conn}(v_j, k, z_{v_j})$ to obtain solution to $T_{v_j}$ for $j = 2, \ldots, d$ .
The complexity of this part is given by $g(n, k/2) \log n + g(n_2, k) + \cdots + g(n_d, k)$.

\medskip

\noindent \textbf{Case 2: Otherwise; i.e., all the light subtrees contain at most $k/2$ components.} \quad 
We call $\textsc{Conn}(v_j, k/2, *)$ to $j = 2, \ldots, d$ sequentially.
The complexity of this part is given by $g(n_2, k/2) + \cdots + g(n_d, k/2)$.

\medskip

The total complexity of the above procedure is given by
\begin{align}
    g(n, k) \le g(n_1, k) + g(n, k/2) (1 + \log n) + g(n_2, k) + \cdots + g(n_d, k) + O(k C).
\end{align}
The maximum is attained at $n_1 = n_2 = n / 2$ and $n_3 = \cdots = n_d = 0$. 
\COMM{TM}{これは嘘．$n_1 = n-1$ and $n_2 = ... = 0$ が最大になりうる．

$n_1$ にガン振りしたとき

$$ g(n-1, k) + g(n, k/2) (1 + \log n) + O(k C) $$

これを解くと明らかに性能落ちますよね
}
\end{comment} 
\end{proof}

\begin{comment}
\begin{remark}
The prefix closeness of the automaton is only required in Theorem~\ref{thm:ksubtree}, i.e., Lemmas~\ref{lem:forallsubtree} and \ref{lem:forallsubtreecomplement} hold on general automata.
\end{remark}
\end{comment}

\begin{example}[$k$ Connected Component Constrained Problem (again)]
By using this technique, the connectivity constrained problem can be solved in $O(n \log n C)$ time, and the $k$ connected component constrained problem can be solved in $O(n^{1 + e} C)$ time for any $e > 0$, since $(\log n)^k = O(n^e)$ for any $e > 0$.
\end{example}

\begin{comment}
\section{Conclusion}

In this study, we considered the automaton constrained tree knapsack problem, and introduced the HLRecDP to obtain a pseudo-polynomial time algorithm with a linear quasi-polynomial factor. 
We also proposed techniques to derive efficient algorithms for the subtree version problems.

Here, we pose two open problems. 
The first problem is characterizing the language accepted by top-down tree automata having bounded diversity of transitions.
These strictly includes the path-closed language, which is represented by deterministic top-down tree automata~\cite{viragh1981deterministic}, and strictly included by monadic second-order language, which is represented by non-deterministic top-down tree automata~\cite{thatcher1968generalized}.
Another possible direction is proving the hardness restricting the linear pseudo-polynomial factor.
\end{comment}

\bibliographystyle{splncs04}% the recommnded bibstyle
\bibliography{main}

\begin{thebibliography}{10}
\providecommand{\url}[1]{\texttt{#1}}
\providecommand{\urlprefix}{URL }
\providecommand{\doi}[1]{https://doi.org/#1}

\bibitem{backurs2017better}
Backurs, A., Indyk, P., Schmidt, L.: Better approximations for tree sparsity in
  nearly-linear time. In: Proceedings of the 28th Annual ACM-SIAM Symposium on
  Discrete Algorithms (SODA'17). pp. 2215--2229 (2017)

\bibitem{bremner2006necklaces}
Bremner, D., Chan, T.M., Demaine, E.D., Erickson, J., Hurtado, F., Iacono, J.,
  Langerman, S., Taslakian, P.: Necklaces, convolutions, and x + y. In:
  Proceedings of the 24th European Symposium on Algorithms (ESA'06). pp.
  160--171 (2006)

\bibitem{chekuri2005recursive}
Chekuri, C., P{\'{a}}l, M.: {A recursive greedy algorithm for walks in directed
  graphs}. In: Proceedings of the 46th Annual Symposium on Foundations of
  Computer Science (FOCS'05). vol.~2005, pp. 245--253 (2005)

\bibitem{cho1997depth}
Cho, G., Shaw, D.X.: A depth-first dynamic programming algorithm for the tree
  knapsack problem. INFORMS Journal on Computing  \textbf{9}(4),  431--438
  (1997)

\bibitem{comon2017tree}
Comon, H., Dauchet, M., Gilleron, R., L\"oding, C., Jacquemard, F., Lugiez, D.,
  Tison, S., Tommasi, M.: Tree automata techniques and applications

\bibitem{cygan2017problems}
Cygan, M., Mucha, M., W{\k{e}}grzycki, K., W{\l}odarczyk, M.: On problems
  equivalent to (min,+)-convolution. arXiv preprint arXiv:1702.07669  (2017)

\bibitem{hirao2013single}
Hirao, T., Yoshida, Y., Nishino, M., Yasuda, N., Nagata, M.: Single-document
  summarization as a tree knapsack problem. In: Proceedings of the 2013
  Conference on Empirical Methods in Natural Language Processing (EMNLP'13).
  pp. 1515--1520 (2013)

\bibitem{hochbaum1994node}
Hochbaum, D.S., Pathria, A.: Node-optimal connected k-subgraphs. Tech. rep.,
  University of California, Berkeley (1994)

\bibitem{ibarra1975fast}
Ibarra, O., Kim, C.: {Fast Approximation Algorithms for the Knapsack and Sum of
  Subset Problems}. Journal of the ACM  \textbf{22}(4),  463--468 (1975)

\bibitem{tomas2018algorithmic}
Ishihata, M., Maehara, T., Rigaux, T.: Algorithmic meta-theorems for monotone
  submodular maximization. arXiv preprint arXiv:1807.04575  (2018)

\bibitem{johnson1983knapsacks}
Johnson, D.S., Niemi, K.: On knapsacks, partitions, and a new dynamic
  programming technique for trees. Mathematics of Operations Research
  \textbf{8}(1),  1--14 (1983)

\bibitem{kellerer2003knapsack}
Kellerer, H., Pferschy, U., Pisinger, D.: Knapsack problems. Springer-Verlag,
  Berlin, Heidelberg (2003)

\bibitem{lawler1977fast}
Lawler, E.L.: {Fast approximation algorithms for knapsack problems}.
  Proceedings of the 18th Annual Symposium on Foundations of Computer Science
  (FOCS'77)  \textbf{4}(4),  339--357 (1977)

\bibitem{lukes1974efficient}
Lukes, J.A.: Efficient algorithm for the partitioning of trees. IBM Journal of
  Research and Development  \textbf{18}(3),  217--224 (1974)

\bibitem{van2006tree}
Van~der Merwe, D., Hattingh, J.M.: Tree knapsack approaches for local access
  network design. European Journal of Operational Research  \textbf{174}(3),
  1968--1978 (2006)

\bibitem{pferschy2009knapsack}
Pferschy, U., Schauer, J.: The knapsack problem with conflict graphs. Journal
  on Graph Algorithms Applications  \textbf{13}(2),  233--249 (2009)

\bibitem{sleator1983data}
Sleator, D.D., Tarjan, R.E.: A data structure for dynamic trees. Journal of
  Computer and System Sciences  \textbf{26}(3),  362--391 (1983)

\bibitem{thatcher1968generalized}
Thatcher, J.W., Wright, J.B.: Generalized finite automata theory with an
  application to a decision problem of second-order logic. Mathematical Systems
  Theory  \textbf{2}(1),  57--81 (1968)

\end{thebibliography}

%\appendix

\end{document}